\documentclass{llncs}

\usepackage{amsmath,amssymb}
\usepackage[utf8]{inputenc}
\usepackage[T1]{fontenc}
\usepackage{mathptmx}
\DeclareMathAlphabet{\mathcal}{OMS}{cmsy}{m}{n}
\usepackage{graphicx}
\usepackage{paralist}
\usepackage{xspace}
\usepackage[english]{babel}
\usepackage{hyperref}

\usepackage{braceMnSymbol}

\usepackage[backgroundcolor=white,linecolor=black,textsize=scriptsize]{todonotes}
\spnewtheorem{fact}{Fact}{\bfseries}{\itshape} %
\spnewtheorem{nrclaim}{Claim}{\itshape}{\rmfamily}

\let\doendproof\endproof
\renewcommand\endproof{~\hfill$\qed$\doendproof}

\newcommand{\low}{\ensuremath{\mathrm{low}}}
\newcommand{\eps}{\varepsilon}

\DeclareMathOperator{\skel}{skel} %
\DeclareMathOperator{\pert}{pert} %
\DeclareMathOperator{\flex}{flex} %
\DeclareMathOperator{\rot}{rot} %
\DeclareMathOperator{\cost}{cost} %
\DeclareMathOperator{\flow}{flow} %

\makeatletter
\DeclareRobustCommand{\bfseries}{%
  \not@math@alphabet\bfseries\mathbf
  \fontseries\bfdefault\selectfont
  \boldmath
}
\makeatother

\title{Orthogonal Graph Drawing with Inflexible Edges}%
\author{Thomas Bläsius, Sebastian Lehmann, Ignaz Rutter}%
\institute{Faculty of Informatics, Karlsruhe Institute of Technology
  (KIT), Germany}

\begin{document}

\pagestyle{headings}

\maketitle

\begin{abstract}

  We consider the problem of creating plane orthogonal drawings of
  \emph{4-planar graphs} (planar graphs with maximum degree~4) with
  constraints on the number of bends per edge.  More precisely, we
  have a \emph{flexibility function} assigning to each edge $e$ a
  natural number $\flex(e)$, its \emph{flexibility}.  The problem
  \textsc{FlexDraw} asks whether there exists an orthogonal drawing
  such that each edge $e$ has at most $\flex(e)$ bends.  It is known
  that \textsc{FlexDraw} is NP-hard if $\flex(e) = 0$ for every edge
  $e$~\cite{gt-curpt-01}.  On the other hand, \textsc{FlexDraw} can be
  solved efficiently if $\flex(e) \ge 1$~\cite{bkrw-ogdfc-12} and is
  trivial if $\flex(e) \ge 2$~\cite{bk-bhogd-98} for every edge $e$.

  To close the gap between the NP-hardness for $\flex(e) = 0$ and the
  efficient algorithm for $\flex(e) \ge 1$, we investigate the
  computational complexity of \textsc{FlexDraw} in case only few edges
  are \emph{inflexible} (i.e., have flexibility~$0$).  We show that
  for any $\eps > 0$ {\sc FlexDraw} is NP-complete for instances with
  $O(n^\eps)$ inflexible edges with pairwise distance
  $\Omega(n^{1-\eps})$ (including the case where they induce a
  matching).  On the other hand, we give an FPT-algorithm with running
  time $O(2^k\cdot n \cdot T_{\flow}(n))$, where $T_{\flow}(n)$ is the
  time necessary to compute a maximum flow in a planar flow network
  with multiple sources and sinks, and $k$ is the number of inflexible
  edges having at least one endpoint of degree~4.
\end{abstract}

\section{Introduction}
\label{sec:introduction}

Bend minimization in orthogonal drawings is a classical problem in the
field of graph drawing.  We consider the following problem called {\sc
  OptimalFlexDraw}.  The input is a 4-planar graph $G$ (from now on
all graphs are 4-planar) together with a cost function $\cost_e\colon
\mathbb N \to \mathbb R\cup \{\infty\}$ assigned to each edge.  We
want to find an orthogonal drawing $\Gamma$ of $G$ such that
$\sum\cost_e(\beta_e)$ is minimal, where $\beta_e$ is the number of
bends of $e$ in $\Gamma$.  The basic underlying decision problem {\sc
  FlexDraw} restricts the cost function of every edge $e$ to
$\cost_e(\beta) = 0$ for $\beta \in [0, \flex(e)]$ and $\cost_e(\beta)
= \infty$ otherwise, and asks whether there exists a \emph{valid}
drawing (i.e., a drawing with finite cost).  The value $\flex(e)$ is
called the \emph{flexibility} of $e$.  Edges with flexibility~0 are
called \emph{inflexible}.

Note that \textsc{FlexDraw} represents the important base case of
testing for the existence of a drawing with cost~0 that is included in
solving \textsc{OptimalFlexDraw}.

Garg and Tamassia~\cite{gt-curpt-01} show that {\sc FlexDraw} is
NP-hard in this generality, by showing that it is NP-hard if every
edge is inflexible.  For special cases, namely planar graphs with
maximum degree~3 and series-parallel graphs, Di~Battista et
al.~\cite{blv-sood-98} give an algorithm minimizing the total number
of bends, which solves {\sc OptimalFlexDraw} with $\cost_e(\beta) =
\beta$ for each edge $e$.  Their approach can be used to solve {\sc
  FlexDraw}, as edges with higher flexibility can be modeled by a path
of inflexible edges. Biedl and Kant~\cite{bk-bhogd-98} show that every
4-planar graph (except for the octahedron) admits an orthogonal
drawing with at most two bends per edge.  Thus, {\sc FlexDraw} is
trivial if the flexibility of every edge is at least~2.  Bläsius et
al.~\cite{bkrw-ogdfc-12,brw-oogd-13} tackle the NP-hard problems {\sc
  FlexDraw} and {\sc OptimalFlexDraw} by not counting the first bend
on every edge.  They give a polynomial time algorithm solving {\sc
  FlexDraw} if the flexibility of every edge is at
least~1~\cite{bkrw-ogdfc-12}.  Moreover, they show how to efficiently
solve {\sc OptimalFlexDraw} if the cost function of every edge is
convex and allows the first bend for free~\cite{brw-oogd-13}.

When restricting the allowed drawings to those with a specific planar
embedding, the problem {\sc OptimalFlexDraw} becomes significantly
easier.  Tamassia~\cite{t-eggmb-87} shows how to find a drawing with
as few bends as possible by computing a flow in a planar flow network.
This flow network directly extends to a solution of {\sc
  OptimalFlexDraw} with fixed planar embedding, if all cost functions
are convex.  Cornelsen and Karrenbauer~\cite{ck-abm-12} recently
showed, that this kind of flow network can be solved in $O(n^{3/2})$
time.

\paragraph{Contribution \& Outline.}

In this work we consider {\sc OptimalFlexDraw} for instances that may
contain inflexible edges, closing the gap between the general
NP-hardness result~\cite{gt-curpt-01} and the polynomial-time
algorithms in the absence of inflexible
edges~\cite{bkrw-ogdfc-12,brw-oogd-13}.  After presenting some
preliminaries in Section~\ref{sec:preliminaries}, we show in
Section~\ref{sec:match-infl-edges} that {\sc FlexDraw} remains NP-hard
even for instances with only $O(n^\eps)$ (for any $\eps > 0$)
inflexible edges that are distributed evenly over the graph, i.e.,
they have pairwise distance $\Omega(n^{1-\eps})$.  This includes the
cases where the inflexible edges are restricted to form very simple
structures such as a matching.

On the positive side, we describe a general algorithm that can be used
to solve {\sc OptimalFlexDraw} by solving smaller subproblems
(Section~\ref{sec:general-algorithm}).  This provides a framework for
the unified description of bend minimization algorithms which covers
both, previous work and results presented in this paper.  We use this
framework in Section~\ref{sec:seri-parall-graphs} to solve {\sc
  OptimalFlexDraw} for series-parallel graphs with monotone cost
functions.  This extends the algorithm of Di~Battista et
al.~\cite{blv-sood-98} by allowing a significantly larger set of cost
functions (in particular, the cost functions may be non-convex).  In
Section~\ref{sec:fpt-algorithm}, we present our main result, which is
an FPT-algorithm with running time~$O(2^k \cdot n \cdot
T_{\flow}(n))$, where $k$ is the number of inflexible edges incident
to degree-4 vertices, and $T_{\flow}(n)$ is the time necessary to
compute a maximum flow in a planar flow network of size $n$ with
multiple sources and sinks.  Note that we can require an arbitrary
number of edges whose endpoints both have degree at most~3 to be
inflexible without increasing the running time.


\section{Preliminaries}
\label{sec:preliminaries}

\subsection{Connectivity \& the Composition of Graphs}
\label{sec:composition-graphs}

A graph $G$ is \emph{connected} if there exists a path between every
pair of vertices.  A \emph{separating $k$-set} $S$ is a subset of
vertices of $G$ such that $G - S$ is not connected.  Separating 1-sets
are called \emph{cutvertices} and separating 2-sets \emph{separation
  pairs}.  A connected graph without cutvertices is \emph{biconnected}
and a biconnected graph without separation pairs is
\emph{triconnected}.  The \emph{blocks} of a connected graph are its
maximal (with respect to inclusion) biconnected subgraphs.

An \emph{$st$-graph} $G$ is a graph with two designated vertices $s$
and $t$ such that $G + st$ is biconnected and planar.  The vertices
$s$ and $t$ are called the \emph{poles} of $G$.  Let $G_1$ and $G_2$
be two $st$-graphs with poles $s_1$, $t_1$ and $s_2$, $t_2$,
respectively.  The \emph{series composition} $G$ of $G_1$ and $G_2$ is
the union of $G_1$ and $G_2$ where $t_1$ is identified with $s_2$.
Clearly, $G$ is again an $st$-graph with the poles $s_1$ and $t_2$.
In the \emph{parallel composition} $G$ of $G_1$ and $G_2$ the vertices
$s_1$ and $s_2$ and the vertices $t_1$ and $t_2$ are identified with
each other and form the poles of $G$.  An $st$-graph is
\emph{series-parallel}, if it is a single edge or the series or
parallel composition of two series-parallel graphs.

To be able to compose all $st$-graphs, we need a third composition.
Let $G_1, \dots, G_\ell$ be a set of $st$-graphs with poles $s_i$ and
$t_i$ associated with $G_i$.  Moreover, let $H$ be an $st$-graph with
poles $s$ and $t$ such that $H + st$ is triconnected and let $e_1,
\dots, e_\ell$ be the edges of $H$.  Then the \emph{rigid composition}
$G$ with respect to the so-called \emph{skeleton} $H$ is obtained by
replacing each edge $e_i$ of $H$ by the graph $G_i$, identifying the
endpoints of $e_i$ with the poles of $G_i$.  It follows from the
theory of SPQR-trees that every $st$-graph is either a single edge or
the series, parallel or rigid composition of
$st$-graphs~\cite{dt-omtc-96,dt-opt-96}.

\subsection{SPQR-Tree}
\label{sec:spqr-tree}

The SPQR-tree $\mathcal T$ of a biconnected $st$-graph $G$ containing
the edge $st$ is a rooted tree encoding series, parallel and rigid
compositions of $st$-graphs that result in the graph
$G$~\cite{dt-omtc-96,dt-opt-96}.  The leaves of $\mathcal T$ are
\emph{Q-nodes} representing the edges of $G$ and thus the $st$-graphs
we start with.  The root of $\mathcal T$ is also a Q-node,
representing the special edge $st$.  Each inner node is either an
\emph{S-node}, representing one or more series compositions of its
children, a \emph{P-node}, representing one or more parallel
compositions of its children, or an \emph{R-node}, representing a
rigid composition of its children.

Recall that the rigid composition is performed with respect to a
skeleton.  For an R-node $\mu$, let $H$ be the skeleton of the
corresponding rigid composition with poles $s_\mu$ and $t_\mu$.  We
call $H + s_{\mu}t_\mu$ the \emph{skeleton} of the $\mu$ and denote it
by $\skel(\mu)$.  The special edge $s_{\mu}t_\mu$ is called
\emph{parent edge}, all other edges are virtual edges, each
corresponding to one child of $\mu$.  We also add skeletons to the
other nodes.  For an S-node $\mu$, the skeleton $\skel(\mu)$ is a path
of virtual edges (one for each child) from $s_\mu$ to~$t_\mu$ together
with the parent edge $s_{\mu}t_\mu$.  The skeleton of a P-node $\mu$
is a bunch of parallel virtual edges (one for each child) between
$s_\mu$ and $t_\mu$ together with the parent edge $s_{\mu}t_\mu$.  The
skeleton of a Q-node contains the edge it represents in $G$ together
with a parallel parent edge.  The root representing $st$ has no parent
edge, thus this additional edge is a virtual edge corresponding to the
unique child of the root.

When not allowing pairs of adjacent S-nodes and pairs of adjacent
P-nodes in $\mathcal T$, then the SPQR-tree is unique for a fixed edge
$st$ in $G$.  Moreover, using the endpoints of a different edge as
poles of $G$ results in the same SPQR-tree with a different root (the
parent edge in each skeleton may also change).  For fixed poles $s$
and $t$, there is a bijection between the planar embeddings of $G$
with $st$ on the outer face and the combinations of embeddings of all
skeletons with their parent edges on the outer face.  The
\emph{pertinent graph} $\pert(\mu)$ of a node $\mu$ of $\mathcal T$ is
recursively defined to be the skeleton $\skel(\mu)$ without the parent
edge $s_{\mu}t_\mu$ after the replacement of every virtual edge with
the pertinent graph of the corresponding child.  Note that the
pertinent graph of the root is $G$ itself. The SPQR-tree can be
computed in linear time~\cite{gm-lti-00}.

\subsection{Orthogonal Representation}
\label{sec:orth-repr}

To handle orthogonal drawings of a graph $G$, we use the abstract
concept of orthogonal representations neglecting distances in a
drawing.  Orthogonal representations were introduced by
Tamassia~\cite{t-eggmb-87}, however, we use a slight modification that
makes it easier to work with, as bends of edges and bends at vertices
are handled the same.  Let $\Gamma$ be a \emph{normalized} orthogonal
drawing of $G$, that is every edge has only bends in one direction.
If additional bends cannot improve the drawing (i.e., costs are
monotonically increasing), a normalized optimal drawing
exists~\cite{t-eggmb-87}.  We assume that all orthogonal drawings we
consider are normalized.

We assume that $G$ is biconnected.  This simplifies the description,
as each edge and vertex has at most one incidence to a face.  For
connected graphs, referring to the incidence of a vertex or an edge
and a face may be ambiguous.  However, it will be always clear from
the context, which incidence is meant.

Let $e$ be an edge in $G$ that has $\beta$ bends in $\Gamma$ and let
$f$ be a face incident to $e$.  We define the \emph{rotation} of $e$
in $f$ to be $\rot(e_f) = \beta$ and $\rot(e_f) = -\beta$ if the bends
of $e$ form $90^\circ$ and $270^\circ$ angles in $f$, respectively.
For a vertex $v$ forming the angle $\alpha$ in the face $f$, we define
$\rot(v_f) = 2 - \alpha/90^\circ$.  
Note that, when traversing a face of $G$ in clockwise
(counter-clockwise for the outer face) direction, the right and left
bends correspond to rotations of $1$ and $-1$, respectively (we may
have two left bends at once at vertices of degree~1).  The values for
the rotations we obtain from a drawing $\Gamma$ satisfy the following
properties; see Fig.~\ref{fig:orthogonal-representation}a.
\begin{compactenum}[(1)]
\item \label{prop:1}The sum over all rotations in a face is $4$ ($-4$
  for the outer face).
\item \label{prop:2}For every edge $e$ with incident faces $f_\ell$
  and $f_r$ we have $\rot(e_{f_\ell}) + \rot(e_{f_r}) = 0$.
\item \label{prop:3}The sum of rotations around a vertex $v$ is
  $2\cdot \deg(v) - 4$.
\item \label{prop:4}The rotations at vertices lie in the range $[-2,
  1]$.
\end{compactenum}
Let $\mathcal R$ be a structure consisting of an embedding of $G$ plus
a set of values fixing the rotation for every vertex-face and
edge-face incidence.  We call $\mathcal R$ an \emph{orthogonal
  representation} of $G$ if the rotation values satisfy the above
properties~\eqref{prop:1}--\eqref{prop:4}.  Given an orthogonal
representation $\mathcal R$, a drawing inducing the specified
rotation values exists and can be computed
efficiently~\cite{t-eggmb-87}.

\paragraph{Orthogonal Representations and Bends of $st$-Graphs.}

We extend the notion of rotation to paths, ; conceptually this is very
similar to spirality~\cite{blv-sood-98}.  Let $\pi$ be a path from
vertex $u$ to vertex $v$.  We define the rotation of $\pi$ (denoted by
$\rot(\pi)$) to be the number of bends to the right minus the number
of bends to the left when traversing $\pi$ from $u$ to $v$.

There are two special paths in an $st$-graph $G$.  Let $s$ and $t$ be
the poles of $G$ and let $\mathcal R$ be an orthogonal representation
with $s$ and $t$ on the outer face.  Then $\pi(s, t)$ denotes the path
from $s$ to $t$ when traversing the outer face of $G$ in
counter-clockwise direction.  Similarly, $\pi(t, s)$ is the path from
$t$ to $s$.  We define the \emph{number of bends} of $\mathcal R$ to
be $\max\{|\rot(\pi(s, t))|, |\rot(\pi(t, s))|\}$.  Note that a single
edge $e = st$ is also an $st$-graph.  Note further that the notions of
the number of bends of the edge $e$ and the number of bends of the
$st$-graph $e$ coincide.  Thus, the above definition is consistent.

When considering orthogonal representations of $st$-graphs, we always
require the poles $s$ and $t$ to be on the outer face.  We say that
the vertex $s$ has \emph{$\sigma$ occupied incidences} if $\rot(s_f) =
\sigma - 3$ where $f$ is the outer face.  We also say that $s$ has $4
- \sigma$ \emph{free incidences} in the outer face.  If the poles $s$
and $t$ have $\sigma$ and $\tau$ occupied incidences in $\mathcal R$,
respectively, we say that $\mathcal R$ is a \emph{$(\sigma,
  \tau)$-orthogonal representation}; see
Fig.~\ref{fig:orthogonal-representation}b.

Note that $\rot(\pi(s, t))$ and $\rot(\pi(t, s))$ together with the
number of occupied incidences $\sigma$ and $\tau$ basically describe
the outer shape of $G$ and thus how it has to be treated if it is a
subgraph of some larger graph.  Using the bends of $\mathcal R$
instead of the rotations of $\pi(s, t)$ and $\pi(t, s)$ implicitly
allows to mirror the orthogonal representation (and thus exchanging
$\pi(s, t)$ and $\pi(t, s)$).

\begin{figure}[tb]
  \centering
  \includegraphics[page=1]{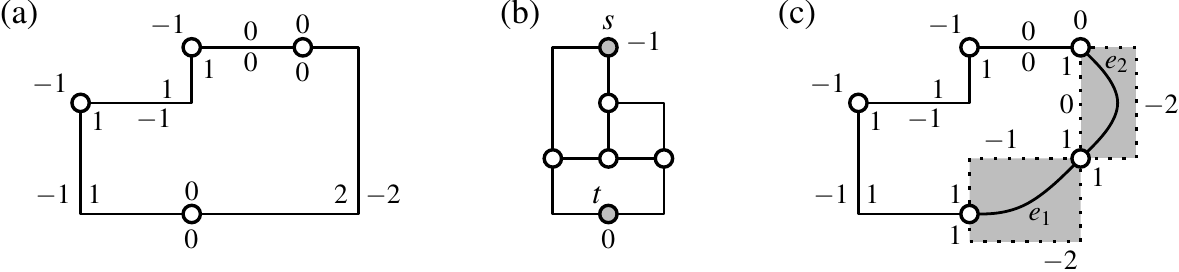}
  \caption{(a)~An orthogonal drawing together with its orthogonal
    representation given by the rotation values.  (b)~A $(2,
    3)$-orthogonal representation ($s$ and $t$ have $2$ and $1$ free
    incidences, respectively).  (c)~An orthogonal representation with
    thick edges $e_1$ and $e_2$.  The gray boxes indicate how many
    attachments the thick edges occupy, i.e., $e_1$ is a $(2, 3)$-edge
    and $e_2$ is a $(2, 2)$-edge.  Both thick edges have two bends.}
  \label{fig:orthogonal-representation}
\end{figure}

\paragraph{Thick Edges.}

In the basic formulation of an orthogonal representation, every edge
\emph{occupies} exactly one incidence at each of its endpoints, that
is an edge enters each of its endpoint from exactly one of four
possible directions.  We introduce \emph{thick edges} that may occupy
more than one incidence at each endpoint to represent larger
subgraphs.

Let $e = st$ be an edge in $G$.  We say that $e$ is a \emph{$(\sigma,
  \tau)$-edge} if $e$ is defined to occupy $\sigma$ and $\tau$
incidences at $s$ and $t$, respectively.  Note that the total amount
of occupied incidences of a vertex in $G$ must not exceed~4.  With
this extended notion of edges, we define a structure $\mathcal R$
consisting of an embedding of $G$ plus a set of values for all
rotations to be an \emph{orthogonal representation} if it satisfies
the following (slightly extended) properties; see
Fig.~\ref{fig:orthogonal-representation}c.
\begin{compactenum}[(1)]
\item \label{prop:1-extended}The sum over all rotations in a face is
  $4$ ($-4$ for the outer face).
\item \label{prop:2-extended}For every $(\sigma, \tau)$-edge $e$ with
  incident faces $f_\ell$ and $f_r$ we have $\rot(e_{f_\ell}) +
  \rot(e_{f_r}) = 2 - (\sigma + \tau)$.
\item \label{prop:3-extended}The sum of rotations around a vertex $v$
  with incident edges $e_1, \dots, e_\ell$ occupying $\sigma_1, \dots,
  \sigma_\ell$ incidences of $v$ is $\sum (\sigma_i + 1) - 4$
\item \label{prop:4-extended}The rotations at vertices lie in the
  range $[-2, 1]$.
\end{compactenum}
Note that requiring every edge to be a $(1, 1)$-edge in this
definition of an orthogonal representation exactly yields the previous
definition without thick edges.  The \emph{number of bends} of a
(thick) edge $e$ incident to the faces $f_\ell$ and $f_r$ is
$\max\{|\rot(e_{f_\ell})|, |\rot(e_{f_r})|\}$.  Unsurprisingly,
replacing a $(\sigma, \tau)$-edge with $\beta$ bends in an orthogonal
representation by a $(\sigma, \tau)$-orthogonal representation with
$\beta$ bends of an arbitrary $st$-graph yields a valid orthogonal
representation~\cite[Lemma~5]{bkrw-ogdfc-12}.

\section{A Matching of Inflexible Edges}
\label{sec:match-infl-edges}

In this section, we show that {\sc FlexDraw} is NP-complete even if
the inflexible edges form a matching.  In fact, we show the stronger
result of NP-hardness of instances with $O(n^\eps)$ inflexible edges
(for $\eps > 0$) even if these edges are distributed evenly over the
graph, that is they have pairwise distance~$\Omega(n^{1-\eps})$.  This
for example shows NP-hardness for instances with $O(\sqrt{n})$
inflexible edges with pairwise distances of $\Omega(\sqrt{n})$.

We adapt the proof of NP-hardness by Garg and
Tamassia~\cite{gt-curpt-01} for the case that all edges of an instance
of {\sc FlexDraw} are inflexible.  For a given instance of
\textsc{Nae-3Sat} (Not All Equal 3SAT) they show how to construct a
graph $G$ that admits an orthogonal representation without bends if
and only if the instance of \textsc{Nae-3Sat} is satisfiable.  The
graph $G$ is obtained by first constructing a graph $F$ that has a
unique planar embedding~\cite[Lemma~5.1]{gt-curpt-01} and replacing
the edges of $F$ by special $st$-graphs, the so called tendrils and
wiggles.  Both, tendrils and wiggles, have degree~1 at both poles and
a unique planar embedding up to possibly a flip.  It follows for each
vertex $v$ of $G$, that the cyclic order of incident edges around $v$
is fixed up to a flip.  This implies the following lemma.

\begin{lemma}[Garg \& Tamassia~\cite{gt-curpt-01}]
  \label{lem:np-hardness-fixed-ordering}
  {\sc FlexDraw} is NP-hard, even if the order of edges around each
  vertex is fixed up to reversal.
\end{lemma}

We assume that our instances do not contain degree-2 vertices; their
incident edges can be replaced by a single edge with higher
flexibility.  In the following, we first show how to replace vertices
of degree~3 by graphs of constant size such that each inflexible edge
is incident to two vertices of degree~4.  Afterwards, we can replace
degree-4 vertices by smaller subgraphs with positive flexibility,
which increases the distance between the inflexible edges.  We start
with the description of an $st$-graph that has either~1 or~2 bends in
every valid orthogonal representation.

The wheel $W_4$ of size~4 consists of a 4-cycles $v_1, \dots, v_4$
together with a center~$u$ connected to each of the vertices $v_1,
\dots, v_4$; see Figure~\ref{fig:hardness-matching}a.  We add the
two vertices $s$ and $t$ together with the inflexible edges $sv_1$ and
$tv_2$ to $W_4$.  Moreover, we set the flexibility of $v_3v_4$
to~2 and the flexibilities of all other edges to~1.  We call the
resulting $st$-graph \emph{bend gadget} and denote it by $B_{1,2}$.
We only consider embeddings of $B_{1, 2}$ where all vertices except
for $u$ lie on the outer face.  Figure~\ref{fig:hardness-matching}(b)
shows two valid orthogonal representations of $B_{1, 2}$, one with~1,
the other with~2 bends.  Clearly, the number of bends cannot be
reduced to~0 (or increased above~2) without violating the flexibility
constraints of edges on the path $\pi(s, t)$ (or on the path $\pi(t,
s)$).  Thus, $B_{1, 2}$ has either~1 or~2 bends in every orthogonal
representation.  Moreover, if its embedding is fixed, then the
direction of the bends is also fixed.

\begin{figure}[tb]
  \centering
  \includegraphics[page=4]{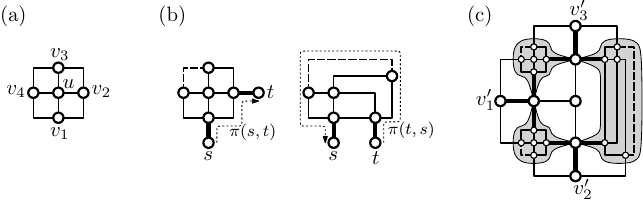}
  \caption{The bold edges are inflexible; dashed edges have
    flexibility~2; all other edges have flexibility~1.  (a)~The
    wheel~$W_4$.  (b)~The bend gadget $B_{1, 2}$. (c)~The gadget
    $W_3'$ for replacing degree-3 vertices.  The marked subgraphs are
    bend gadgets.}
  \label{fig:hardness-matching}
\end{figure}

We now use the bend gadget as building block for a larger gadget.  We
start with the wheel $W_3$ of size~3 consisting of a triangle $v_1,
v_2, v_3$ together with a center $u$ connected to $v_1$, $v_2$, and
$v_3$.  The flexibilities of the edges incident to the center are set
to~1, each edge in the triangle is replaced by a bend gadget $B_{1,
  2}$.  To fix the embedding of the bend gadgets, we add three
vertices $v_1'$, $v_2'$, and $v_3'$ connected with inflexible edges to
$v_1$, $v_2$, and $v_3$, respectively, and connect them to the free
incidences in the bend gadgets, as shown in
Figure~\ref{fig:hardness-matching}(c).  We denote the resulting graph
by $W_3'$.  Clearly, in the cycle of bend gadgets, two of them have
one bend and the other has two bends in every valid orthogonal
representation of $W_3'$.  Thus, replacing a vertex $v$ with incident
edges $e_1$, $e_2$, and $e_3$ by $W_3'$, attaching the edge $e_i$ to
$v_i'$, yields an equivalent instance of \textsc{FlexDraw}.  Note that
such a replacement increases the degree of one incidence of $e_1$,
$e_2$, and $e_3$ form~3 to~4.  Moreover, every inflexible edge
contained in $W_3'$ is incident to two vertices of degree~4.  We
obtain the following lemma.

\begin{lemma}
  \label{lem:np-hardness-fixde-ordering-deg-4}
  {\sc FlexDraw} is NP-hard, even if the endpoints of each inflexible
  edge have degree~4 and if the order of edges around each vertex is
  fixed up to reversal.
\end{lemma}
\begin{proof}
  Let $G$ be an instance of {\sc FlexDraw} such that the order of
  edges around each vertex is fixed up to reversal.  As {\sc FlexDraw}
  restricted to these kinds of instances is NP-hard, due to
  Lemma~\ref{lem:np-hardness-fixed-ordering}, it suffices to find an
  equivalent instance where additionally the endpoints of each
  inflexible have degree~4.  Pairs of edges incident to a vertex of
  degree~2 can be simply replaced by an edge with higher flexibility.
  Thus, we can assume that every vertex in $G$ has degree~3 or
  degree~4.  Replacing every degree-3 vertex incident to an inflexible
  edge by the subgraph $W_3'$ described above clearly leads to an
  equivalent instance with the desired properties.
\end{proof}

Similar to the replacement of degree-3 vertices by $W_3'$, we can
replace degree-4 vertices by the wheel $W_4$, setting the flexibility
of every edge of $W_4$ to~1.  It is easy to see, that every valid
orthogonal representation of $W_4$ has the same outer shape, that is a
rectangle, with one of the vertices $v_1, \dots, v_4$ on each side;
see Figure~\ref{fig:hardness-matching}(a).  Thus, replacing a vertex
$v$ with incident edges $e_1, \dots, e_4$ (in this order) by $W_4$,
attaching $e_1, \dots, e_4$ to the vertices $v_1, \dots, v_4$ yields
an equivalent instance of {\sc FlexDraw}.  We obtain the following
theorem.

\begin{theorem}
  \textsc{FlexDraw} is NP-complete even for instances of size $n$ with
  $O(n^\eps)$ inflexible edges with pairwise distance
  $\Omega(n^{1-\eps})$.
\end{theorem}
\begin{proof}
  As \textsc{FlexDraw} is clearly in NP, it remains to show
  NP-hardness.  Let $G$ be the instance of \textsc{FlexDraw} such that
  the endpoints of each inflexible edge have degree~4 and such that
  the order of edges around each edge is fixed up to reversal.  {\sc
    FlexDraw} restricted to these kinds of instances is NP-hard due to
  Lemma~\ref{lem:np-hardness-fixde-ordering-deg-4}.  We show how to
  build an equivalent instance with $O(n^\eps)$ inflexible edges with
  pairwise distance $\Omega(n^{1-\eps})$ for any $\eps > 0$.

  Let $e$ be an inflexible edge in $G$ with incident vertices $u$ and
  $v$, which both have degree~4.  Replacing each of the vertices $u$
  and $v$ by the wheel~$W_4$ yields an equivalent instance of
  \textsc{FlexDraw} and the distance of $e$ to every other inflexible
  edge is increased by a constant.  Note that this does not increase
  the number of inflexible edges.  Let $n_G$ be the number of vertices
  in $G$.  Applying this replacement $n_G^{1/\eps-1}$ times to the
  vertices incident to each inflexible edge yields an equivalent
  instance $G'$.  In $G'$ every pair of inflexible edges has distance
  $\Omega(n_G^{1/\eps-1})$.  Moreover, $G'$ has size $O(n_G^{1/\eps})$,
  as we have $n_G$ inflexible edges.  Substituting $n_G^{1/\eps}$ by
  $n$ shows that we get an instance of size $n$ with $O(n^\eps)$
  inflexible edges with pairwise distance $\Omega(n^{1-\eps})$.
\end{proof}

Note that the instances described above may contain edges with
flexibility larger than~1.  We can get rid of that as follows.  An
edge $e$ with flexibility~$\flex(e) > 0$ can have the same numbers of
bends like the $st$-graph consisting of the wheel $W_4$
(Figure~\ref{fig:hardness-matching}(a)) with the additional edges
$sv_1$ with $\flex(sv_1) = 1$ and $tv_3$ with $\flex(tv_3) = \flex(e)
- 1$.  Thus, we can successively replace edges with rotation above~1
by these kinds of subgraphs, leading to an equivalent instance where
all edges have flexibility~1 or~0.

\section{The General Algorithm}
\label{sec:general-algorithm}

In this section we describe a general algorithm that can be used to
solve {\sc OptimalFlexDraw} by solving smaller subproblems for the
different types of graph compositions.  To this end, we start with the
definition of cost functions for subgraphs, which is straightforward.
The \emph{cost function} $\cost(\cdot)$ of an $st$-graph $G$ is
defined such that $\cost(\beta)$ is the minimum cost of all orthogonal
representations of $G$ with $\beta$ bends.  The \emph{$(\sigma,
  \tau)$-cost function} $\cost^\sigma_\tau(\cdot)$ of $G$ is defined
analogously by setting $\cost^\sigma_\tau(\beta)$ to the minimum cost
of all $(\sigma, \tau)$-orthogonal representations of $G$ with $\beta$
bends.  Clearly, $\sigma, \tau \in \{1, \dots 4\}$, though, for a
fixed graph $G$, not all values may be possible.  If for example
$\deg(s) = 1$, then $\sigma$ is~1 for every orthogonal representation
of $G$.  Note that there is a lower bound on the number of bends
depending on $\sigma$ and $\tau$.  For example, a $(2, 2)$-orthogonal
representation has at least one bend and thus $\cost_2^2(0)$ is
undefined.  We formally set undefined values to $\infty$.

With \emph{the cost functions} of $G$ we refer to the collection of
$(\sigma, \tau)$-cost functions of $G$ for all possible combinations
of $\sigma$ and $\tau$.  Let $G$ be the composition of two or more
(for a rigid composition) graphs $G_1, \dots, G_\ell$.  Computing the
cost functions of $G$ assuming that the cost functions of $G_1, \dots,
G_\ell$ are known is called \emph{computing cost functions of a
  composition}.  The following theorem states that the ability to
compute cost functions of compositions suffices to solve {\sc
  OptimalFlexDraw}.  The terms $T_S$, $T_P$ and $T_R(\ell)$ denote the
time necessary to compute the cost functions of a series, a parallel,
and a rigid composition with skeleton of size $\ell$, respectively.

\begin{theorem}
  \label{thm:general-algo-biconn}
  Let $G$ be an $st$-graph containing the edge $st$.  An optimal
  $(\sigma, \tau)$-orthogonal representation of $G$ with $st$ on the
  outer face can be computed in $O(n T_S + n T_P + T_R(n))$
  time.
\end{theorem}
\begin{proof}
  Let $\mathcal T$ be the SPQR-tree of $G$.  To compute an optimal
  orthogonal representation of $G$ with $st$ on the outer face, we
  root $\mathcal T$ at the Q-node corresponding to $st$ and traverse
  it bottom up.  When processing a node $\mu$, we compute the cost
  functions of $\pert(\mu)$, which finally (in the root) yields the
  cost functions of the $st$-graph $G$ and thus optimal $(\sigma,
  \tau)$-orthogonal representations (for all possible values of
  $\sigma$ and $\tau$) with $st$ on the outer face.

  If $\mu$ is a \textbf{Q-node} but not the root, then $\pert(\mu)$ is
  an edge and the cost function of this edge is given with the input.

  If $\mu$ is an \textbf{S-node}, its pertinent graph can be obtained
  by applying multiple series compositions.  Since the skeleton of an
  S-node leaves no embedding choice, we can compute the cost function
  of $\pert(\mu)$ by successively computing the cost functions of the
  compositions, which takes $O(|\skel(\mu)|\cdot T_S)$ time.

  If $\mu$ is a \textbf{P-node}, then $\pert(\mu)$ can be obtained by
  applying multiple parallel compositions.  In contrast to S-nodes the
  skeleton of a P-node leaves an embedding choice, namely changing the
  order of the parallel edges.  As composing the pertinent graphs of
  the children of $\mu$ in a specific order restricts the embedding of
  $\skel(\mu)$, we cannot apply the compositions in an arbitrary order
  if $\skel(\mu)$ contains more than two parallel edges (not counting
  the parent edge).  However, since $\skel(\mu)$ contains at most
  three parallel edges (due to the restriction to degree~4), we can
  try all composition orders and take the minimum over the resulting
  cost functions.  As there are only constantly many orders and for
  each order a constant number of compositions is performed, computing
  the cost function of $\pert(\mu)$ takes $O(T_P)$ time.

  If $\mu$ is an \textbf{R-node}, the pertinent graph of $\mu$ is the
  rigid composition of the pertinent graphs of its children with
  respect to the skeleton $\skel(\mu)$.  Thus, the cost functions of
  $\pert(\mu)$ can be computed in $O(T_R(|\skel(\mu)|))$ time.
  
  If $\mu$ is the \textbf{root}, that is the Q-node corresponding to
  $st$, then $\pert(\mu) = G$ is a parallel composition of the
  pertinent graph of the child of $\mu$ and the edge $st$ and thus its
  cost function can be computed in $O(T_P)$ time.

  As the total size of S-node skeletons, the number of P-nodes and the
  total size of R-node skeletons is linear in the size of $G$, the
  running time is in $O(n \cdot T_S + n \cdot T_P + T_R(n))$.
\end{proof}

Applying Theorem~\ref{thm:general-algo-biconn} for each pair of
adjacent nodes as poles in a given instance of {\sc OptimalFlexDraw}
yields the following corollary.

\begin{corollary}
  \label{cor:general-algo-biconn}
  {\sc OptimalFlexDraw} can be solved in $O(n \cdot (n T_S + n T_P +
  T_R(n)))$ time for biconnected graphs.
\end{corollary}

In the following, we extend this result to the case where $G$ may
contain cutvertices.  The extension is straightforward, however, there
is one pitfall.  Given two blocks $B_1$ and $B_2$ sharing a cutvertex
$v$ such that $v$ has degree~2 in $B_1$ and $B_2$, we have to ensure
for both blocks that $v$ does not form an angle of $180^\circ$.
Thus, for a given graph $G$, we get for each block a list of vertices
and we restrict the set of all orthogonal representations of $G$ to
those where these vertices form~$90^\circ$ angles.  We call these
orthogonal representations \emph{restricted orthogonal
  representations}.  Moreover, we call the resulting cost functions
\emph{restricted cost functions}.  We use the terms $T_S^r$, $T_P^r$
and $T_R^r(\ell)$ to denote the time necessary to compute the
\emph{restricted} cost functions of a series, a parallel, and a rigid
composition, respectively.  We get the following extension of the
previous results.

\begin{theorem}
  \label{thm:general-algo}
  {\sc OptimalFlexDraw} can be solved in $O(n \cdot (n T_S^r + n T_P^r
  + T_R^r(n)))$ time.
\end{theorem}
\begin{proof}
  Let $G$ be an instance of {\sc OptimalFlexDraw}.  We use the BC-tree
  (Block--Cutvertex Tree) of $G$ to represent all possible ways of
  combining embeddings of the blocks of $G$ to an embedding of $G$.
  The \emph{BC-tree} $\mathcal T$ of $G$ contains a B-node for each
  block of $G$, a C-node for each cutvertex of $G$ and an edge between
  a C-node and a B-node if and only if the corresponding cutvertex is
  contained in the corresponding block, respectively.

  Rooting $\mathcal T$ at some B-node restricts the embeddings of the
  blocks as follows.  Let $\mu$ be a B-node (but not the root)
  corresponding to a block $B$ and let $v$ be the cutvertex
  corresponding to the parent of $\mu$.  Then the embedding of $B$ is
  required to have $v$ on its outer face.  It is easy to see that
  every embedding of $G$ is such a restricted embedding with respect
  to some root of $\mathcal T$.  Thus, it suffices to consider each
  B-node of $\mathcal T$ as root and restrict the embeddings as
  described above.

  Before we deal with the BC-tree $\mathcal T$, we preprocess each
  block $B$ of $G$.  Let $v$ be a cutvertex of $B$.  For an edge $e$
  incident to $v$, we can use Theorem~\ref{thm:general-algo-biconn} to
  compute an optimal orthogonal representation of $B$ with $e$ on the
  outer face in $O(n\cdot T_S + n \cdot T_P + T_R(n))$ time.  Since
  ever orthogonal representation with $v$ on the outer face has one of
  its incident edges on the outer face, we can simply force each of
  these edges to the outer face once, to get an optimal orthogonal
  representation of $B$ with $v$ on the outer face.  Clearly, using
  the computation of restricted cost functions yields an optimal
  restricted orthogonal representation.  Doing this for each block of
  $G$ and for each cutvertex in this block leads to a total running
  time of $O(n\cdot(n\cdot T_S + n \cdot T_P + T_R(n)))$.  Moreover,
  we can compute an optimal restricted orthogonal representation of
  each block (without forcing a vertex to the outer face) with the
  same running time (Corollary~\ref{cor:general-algo-biconn}).

  To compute an optimal orthogonal representation of $G$ we choose
  every B-node of the BC-tree $\mathcal T$ as root and consider for
  the block corresponding to the root the optimal orthogonal
  representation (without forcing vertices to the outer face).  For
  all other blocks we consider the optimal orthogonal representation
  with the cutvertex corresponding to its parent on the outer face.
  Note that these orthogonal representations can be easily combined to
  an orthogonal representation of the whole graph, as we enforce
  angles of $90^\circ$ at vertices of degree~2, if they have degree~2
  in another block.  The minimum over all roots leads to an optimal
  orthogonal representation.  As computing this minimum takes $O(n^2)$
  time, it is dominated by the running time necessary to compute the
  orthogonal representation of the blocks.
\end{proof}

Note that Theorem~\ref{thm:general-algo} provides a framework for
uniform treatment of bend minimization over all planar embeddings in
orthogonal drawings.  In particular, the polynomial-time algorithm for
{\sc FlexDraw} with positive flexibility~\cite{bkrw-ogdfc-12} can be
expressed in this way.  There, all resulting cost functions of
$st$-graphs are~0 on a non-empty interval containing~0 (with one minor
exception) and~$\infty$, otherwise.  Thus, the cost functions of the
compositions can be computed using Tamassia's flow network.  The
results on {\sc OptimalFlexDraw}~\cite{brw-oogd-13} can be expressed
similarly.  When restricting the number of bends of each $st$-graph
occurring in the composition to~3, all resulting cost functions are
convex (with one minor exception).  Thus, Tamassia's flow network can
again be used to compute the cost functions of the compositions.  The
overall optimality follows from the fact that there exists an optimal
solution that can be composed in such a way.  In the following
sections we see two further applications of this framework, resulting
in efficient algorithms.

\section{Series-Parallel Graphs}
\label{sec:seri-parall-graphs}

In this section we show that the cost functions of a series
composition (Lemma~\ref{lem:general-series-composition}) and a
parallel composition (Lemma~\ref{lem:general-parallel-composition})
can be computed efficiently.  Using our framework, this leads to a
polynomial-time algorithm for {\sc OptimalFlexDraw} for
series-parallel graphs with monotone cost functions
(Theorem~\ref{thm:sp}).  We note that this is only a slight extension
to the results by Di Battista et al.~\cite{blv-sood-98}.  However, it
shows the easy applicability of the above framework before diving into
the more complicated FPT-algorithm in the following section.

\begin{lemma}
  \label{lem:general-series-composition}
  If the (restricted) cost functions of two $st$-graphs are $\infty$
  for bend numbers larger than $\ell$, the (restricted) cost functions
  of their series composition can be computed in $O(\ell^2)$ time.
\end{lemma}
\begin{proof}
  We first consider the case of non-restricted cost functions.  Let
  $G_1$ and $G_2$ be the two $st$-graphs with poles $s_1$, $t_1$ and
  $s_2$, $t_2$, respectively, and let $G$ be their series composition
  with poles $s = s_1$ and $t = t_2$.  For each of the constantly many
  valid combinations of $\sigma$ and $\tau$, we compute the $(\sigma,
  \tau)$-cost function separately.  Assume for the following, that
  $\sigma$ and $\tau$ are fixed.  Since~$G_1$ and~$G_2$ both have at
  most~$\ell$ bends, $G$ can only have~$O(\ell)$ possible values for
  the number of bends~$\beta$.  We fix the value $\beta$ and show how
  to compute $\cost^\sigma_\tau(\beta)$ in $O(\ell)$ time.

  Let $\mathcal R$ be a $(\sigma, \tau)$-orthogonal representation
  with $\beta$ bends and let $\mathcal R_1$ and $\mathcal R_2$ be the
  $(\sigma_1, \tau_1)$- and $(\sigma_2, \tau_2)$-orthogonal
  representations induced for $G_1$ and $G_2$, respectively.
  Obviously, $\sigma_1 = \sigma$ and $\tau_2 = \tau$ holds.  However,
  there are the following other parameters that may vary (although
  they may restrict each other).  The parameters $\tau_1$ and
  $\sigma_2$; the number of bends $\beta_1$ and $\beta_2$ of $\mathcal
  R_1$ and $\mathcal R_2$, respectively; the possibility that for $i
  \in \{1, 2\}$ the number of bends of $\mathcal R_i$ are determined
  by $\pi(s_i, t_i)$ or by $\pi(t_i, s_i)$, that is $\beta_i =
  -\rot(\pi(s_i, t_i))$ or $\beta_i = -\rot(\pi(t_i, s_i))$; and
  finally, the rotations at the vertex $v$ in the outer face, where
  $v$ is the vertex of $G$ belonging to both, $G_1$ and $G_2$.

  Assume we fixed the parameters $\tau_1$ and $\sigma_2$, the choice
  by which paths $\beta_1$ and $\beta_2$ are determined, the rotations
  at the vertex $v$, and the number of bends $\beta_1$ of $\mathcal
  R_1$.  Then there is no choice left for the number of bends
  $\beta_2$ of $\mathcal R_2$, as choosing a different value for
  $\beta_2$ also changes the number of bends $\beta$ of $G$, which was
  assumed to be fixed.  As each of the parameters can have only a
  constant number of values except for $\beta_1$, which can have
  $O(\ell)$ different values, there are only $O(\ell)$ possible
  choices in total.  For each of these choices, we get a $(\sigma,
  \tau)$-orthogonal representation of $G$ with $\beta$ bends and cost
  $\cost^{\sigma_1}_{\tau_1}(\beta_1) +
  \cost^{\sigma_2}_{\tau_2}(\beta_2)$.  By taking the minimum cost
  over all these choices we get the desired value
  $\cost^\sigma_\tau(\beta)$ in $O(\ell)$ time.

  If we consider restricted cost functions, it may happen that the
  vertex $v$ has degree~2.  Then we need to enforce an angle of
  $90^\circ$ there.  Obviously, this constraint can be easily added to
  the described algorithm.
\end{proof}

\begin{lemma}
  \label{lem:general-parallel-composition}
  If the (restricted) cost functions of two $st$-graphs are $\infty$
  for bend numbers larger than $\ell$, the (restricted) cost functions
  of their parallel composition can be computed in $O(\ell)$ time.
\end{lemma}
\begin{proof}
  If the composition $G$ of $G_1$ and $G_2$ has $\beta$ bends, either
  the graph $G_1$ or the graph $G_2$ also has $\beta$ bends.  Thus,
  the cost function of $G$ is $\infty$ for bend numbers larger than
  $\ell$.  Let the number of bends of $G$ be fixed to $\beta$.
  Similar to the proof of Lemma~\ref{lem:general-series-composition},
  there are the following parameters.  The number of bends $\beta_1$
  and $\beta_2$ of $G_1$ and $G_2$; $\sigma_i$ and $\tau_i$ for $i \in
  \{1, 2\}$; $\sigma$ and $\tau$; the order of the two graphs; and the
  decision whether $\pi(s, t)$ or $\pi(t, s)$ determines the number of
  bends of $G$.  All parameters except for $\beta_1$ and $\beta_2$
  have $O(1)$ possible values.  As mentioned before, we have $\beta =
  \beta_1$ or $\beta = \beta_2$.  In the former case, fixing all
  parameters except for $\beta_2$ leaves no choice for $\beta_2$.  The
  case of $\beta = \beta_2$ leaves no choice for $\beta_1$.  Thus,
  each of the $O(\ell)$ values can be computed in $O(1)$ time, which
  concludes the proof.
\end{proof}

\begin{theorem}
  \label{thm:sp}
  For series-parallel graphs with monotone cost functions {\sc
    OptimalFlexDraw} can be solved in $O(n^4)$ time.
\end{theorem}
\begin{proof}
  To solve {\sc OptimalFlexDraw}, we use
  Theorem~\ref{thm:general-algo}.  As the graphs we consider here are
  series parallel, it suffices to give algorithms that compute the
  cost functions of series and parallel compositions.  Applying
  Lemma~\ref{lem:general-series-composition} and
  Lemma~\ref{lem:general-parallel-composition} gives us running times
  $T_S \in O(\ell^2)$ and $T_P \in O(\ell)$ for these compositions.
  In the following, we show that it suffices to compute the cost
  functions for a linear number of bends, leading to running times
  $T_s \in O(n^2)$ and $T_P \in O(n)$.  Together with the time stated
  by Theorem~\ref{thm:general-algo}, this gives us a total running
  time of $O(n^4)$.

  Let $G$ be an $st$-graph with monotone cost functions assigned to
  the edges.  We show the existence of an optimal orthogonal
  representation of $G$ such that every split component of $G$ has
  $O(n)$ bends.  To this end, consider the flow network $N$ introduced
  by Tamassia~\cite{t-eggmb-87} and let $d$ be the total demand of all
  its sinks.  Let $\mathcal R$ be an optimal orthogonal representation
  of $G$ such that a split component $H$ has at least $d + 1$ bends.
  Then one of the two faces incident to edges of $H$ and to edges of
  $G - H$ has at least $d + 1$ units of outgoing flow.  As the total
  demand of sinks in the flow network is only $d$, there must exist a
  directed cycle $C$ in $N$ such that the flow on each of the arcs in
  $C$ is at least~1.  Reducing the flow on $C$ by~1 yields a new
  orthogonal representation and as the number of bends on no edge is
  increased, the cost does not increase.  As in every step the total
  amount of flow is decreased, the process stops after finitely many
  steps.  The result is an optimal orthogonal representation of $G$
  such that each split component has at most $d$ bends.  Thus, we can
  restrict our search to orthogonal representations in which each
  split component has only up to $d$ bends.  This can be done by
  implicitly setting the costs to $\infty$ for larger values than $d$.
  This concludes the proof, as $d \in O(n)$ holds.
\end{proof}

\section{An FPT-Algorithm for General Graphs}
\label{sec:fpt-algorithm}

Let $G$ be an instance of {\sc FlexDraw}.  We call an edge in $G$
\emph{critical} if it is inflexible and at least one of its endpoints
has degree~4.  We call the instance $G$ of {\sc FlexDraw}
$k$-critical, if it contains exactly $k$ critical edges.  An
inflexible edge that is not critical is \emph{semi-critical}.  The
poles $s$ and $t$ of an $st$-graph~$G$ are considered to have
additional neighbors (which comes from the fact that we usually
consider $st$-graphs to be subgraphs of larger graphs).  More
precisely, inflexible edges incident to the pole~$s$ (or $t$) are
already \emph{critical} if $\deg(s) \ge 2$ (or $\deg(t) \ge 2$).
In the following, we first study cost functions of $k$-critical
$st$-graphs.  Afterwards, we show how to use the insights we got to
give an FPT-algorithm for $k$-critical instances of {\sc FlexDraw}.

\subsection{The Cost Functions of $k$-Critical Instances}
\label{sec:cost-functions-k}

%
%
\begin{figure}[tb]
  \centering
  \includegraphics[page=2]{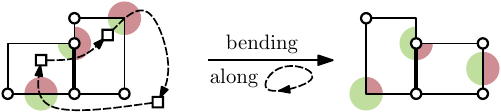}
  \caption{An orthogonal representation (the bold edge is inflexible,
    other edges have flexibility~1), together with a valid cycle
    (dashed).  Bending along this cycle increases the green and
    decreases the red angles.  The resulting orthogonal representation
    is shown on the right.}
  \label{fig:bending-along-cycle}
\end{figure}
Let $G$ be an $st$-graph and let $\mathcal R$ be a valid orthogonal
representation of $G$.  We define an operation that transforms
$\mathcal R$ into another valid orthogonal representation of $G$.  Let
$G^\star$ be the \emph{double directed} dual graph of $G$, that is
each edge $e$ of $G$ with incident faces $g$ and $f$ corresponds to
the two dual edges $(g, f)$ and $(f, g)$.  We call a dual edge
$e^\star = (g, f)$ of $e$ \emph{valid} if one of the following
conditions holds.
\begin{compactenum}[(I)]
\item \label{itm:valid-1}$\rot(e_f) < \flex(e)$ (which is equivalent
  to $-\rot(e_g) < \flex(e)$).
\item \label{itm:valid-2}$\rot(v_f) < 1$ where $v$ is an endpoint of
  $e$ but not a pole.
\end{compactenum}
A simple directed cycle $C^\star$ in $G^\star$ consisting of valid
edges is called \emph{valid cycle}.  Then \emph{bending along
  $C^\star$} changes the orthogonal representation $\mathcal R$ as
follows; see Fig.~\ref{fig:bending-along-cycle}.  Let $e^\star = (g,
f)$ be an edge in $C^\star$ with primal edge $e$.  If $e^\star$ is
valid due to Condition~\eqref{itm:valid-1}, we reduce $\rot(e_g)$ by~1
and increase $\rot(e_f)$ by~1.  Otherwise, if
Condition~\eqref{itm:valid-2} holds, we reduce $\rot(v_g)$ by~1 and
increase $\rot(v_f)$ by~1, where $v$ is the vertex incident to $e$
with $\rot(v_f) < 1$.

\begin{lemma}
  Let $G$ be an $st$-graph with a valid $(\sigma, \tau)$-orthogonal
  representation~$\mathcal R$.  Bending along a valid cycle $C^\star$
  yields a valid $(\sigma, \tau)$-orthogonal representation.
\end{lemma}
\begin{proof}
  First, we show that the resulting rotations still describe an
  orthogonal representation.  Afterwards, we show that this orthogonal
  representation is also valid and that it is a $(\sigma,
  \tau)$-orthogonal representation.  Let $e^\star = (g, f)$ be an edge
  in $C^\star$ with primal edge $e$.  If Condition~\eqref{itm:valid-1}
  holds, then $\rot(e_g)$ is reduced by~1 and $\rot(e_f)$ is increased
  by~1 and thus $\rot(e_g) = -\rot(e_f)$ remains true.  Otherwise,
  Condition~\eqref{itm:valid-2} holds and thus $\rot(v_g)$ is reduced
  by~1 and $\rot(v_f)$ is increased by~1.  The total rotation around
  $v$ does obviously not change.  Moreover, both rotations remain in
  the interval $[-1, 1]$.  Finally, the incoming arc to a face $f$ in
  $C^\star$ increases the rotation around $f$ by~1 and the outgoing
  arc decreases it by~1.  Thus, the total rotation around each face
  remains as it was.

  It remains to show that the resulting orthogonal representation is a
  valid $(\sigma, \tau)$-orthogonal representation.  First,
  Condition~\eqref{itm:valid-1} ensures that we never increase the
  number of bends of an edge $e$ above $\flex(e)$.  Moreover, due to
  the exception in Condition~\eqref{itm:valid-2} where $v$ is one of
  the poles, we never change the rotation of one of the poles.  Thus
  the number of free incidences to the outer face are not changed.
\end{proof}

As mentioned in Section~\ref{sec:general-algorithm}, depending on
$\sigma$ and $\tau$, there is a lower bound on the number of bends of
$(\sigma, \tau)$-orthogonal representations.  We denote this lower
bound by $\beta_\low$; see Fig.~\ref{fig:bends-lower-bound}.

\begin{fact}
  \label{fact:lower-bound}
  A $(\sigma, \tau)$-orthogonal representation has at least
  $\displaystyle \beta_{\low} = \left\lceil\frac{\sigma +
      \tau}{2}\right\rceil - 1$~bends.  
\end{fact}

\begin{figure}[tb]
  \centering
  \includegraphics[page=1]{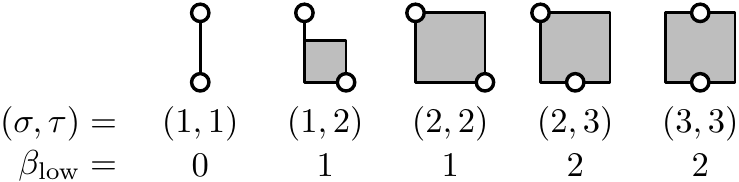}
  \caption{Illustration of Fact~\ref{fact:lower-bound} for some values
    of $\sigma$ and $\tau$.}
  \label{fig:bends-lower-bound}
\end{figure}

For a valid orthogonal representation with a large number of bends,
the following lemma states that we can reduce its bends by bending
along a valid cycle.  This can later be used to show that the cost
function of an $st$-graph is~0 on a significantly large interval.  Or
in other words, arbitrary alterations of cost~0 and cost~$\infty$ that
are hard to handle only occur on a small interval (depending on $k$).
The lemma and its proof are a generalization of Lemma~1
from~\cite{bkrw-ogdfc-12} that incorporates inflexible edges.  For
$\sigma = \tau = 3$ a slightly weaker result holds.

\begin{lemma}
  \label{lem:existence-of-cycle}
  Let $G$ be a $k$-critical $st$-graph and let $\mathcal R$ be a
  valid $(\sigma, \tau)$-orthogonal representation with $\sigma + \tau
  \le 5$.  If $-\rot(\pi(t, s)) \ge \beta_{\low} + k + 1$ holds, then
  there exists a valid cycle $C^\star$ such that bending $\mathcal R$
  along $C^\star$ reduces $-\rot(\pi(t, s))$ by~1.
\end{lemma}
\begin{proof}
  We show the existence of a valid cycle $C^\star$ such that $s$ and
  $t$ lie to the left and right of $C^\star$, respectively.
  Obviously, such a cycle must contain the outer face.  The edge in
  $C^\star$ having the outer face as target ensures that the rotation
  of an edge or a vertex of $\pi(t, s)$ is increased by~1 (which is
  the same as reducing $-\rot(\pi(t, s))$ by~1), where this vertex is
  neither $s$ nor $t$ (due to the exception of
  Condition~\eqref{itm:valid-2}).  Thus, $\rot(\pi(t, s))$ is
  increased by~1 when bending along $C^\star$ and thus $C^\star$ is
  the desired cycle.  We first show the following claim.

  \begin{nrclaim}
    \label{claim-1}
    There exists a valid edge $e^\star$ that either has the outer face
    as source and corresponds to a primal edge $e$ on the path $\pi(s,
    t)$, or is a loop with $s$ to its left and $t$ to right.
  \end{nrclaim}
  
  Assume the claimed edge $e^\star$ does not exist.  We first show
  that the following inequality follows from this assumption.
  Afterwards, we show that this leads to a contradiction to the
  inequality in the statement of the lemma.
  \begin{equation}
    \label{eq:rot-pi-st}
    \rot(\pi(s, t)) \le 
    \begin{cases}
      k, &\text{if } \deg(s) = \deg(t) = 1 \\
      k-1, &\text{otherwise}
    \end{cases}
  \end{equation}
  We first show this inequality for the case where we have \textbf{no
    critical and no semi-critical edges, in particular $k=0$}.  We
  consider the rotation of edges and vertices on $\pi(s, t)$ in the
  outer face $g$.  If an edge or vertex has two incidences to $g$, we
  implicitly consider the incidence corresponding to $\pi(s, t)$.
  Recall that the rotation along $\pi(s, t)$ is the sum over the
  rotations of its edges and of its internal vertices.  The rotation
  of every edge $e$ is $\rot(e_g) = -\flex(e)$ as otherwise $e^\star =
  (g, f)$ would be a valid edge due to Condition~\eqref{itm:valid-1}.
  At an internal vertex $v$ we obviously have $\rot(v_g) \le 1$, as
  larger rotations are not possible at vertices.  Hence, as the
  flexibility of every edge is at least~1 and we have an internal
  vertex less than we have edges, we get $\rot(\pi(s, t)) \le -1$ and
  thus Equation~\eqref{eq:rot-pi-st} is satisfied.

  Next, we allow \textbf{semi-critical edges, but no critical edges
    ($k = 0$ remains)}.  If $\pi(s, t)$ contains a semi-critical edge,
  it has a rotation of~0 (instead of $-1$ for normal edges).  Note
  that we still assume that there is no critical edge in~$\pi(s,t)$,
  i.e., $k=0$.  Moreover, if an internal vertex $v$ is incident to a
  semi-critical edge, it cannot have degree~4.  In this case, there
  must be a face incident to $v$ such that $v$ has rotation at most~0
  in this face.  If this face was not $g$,
  Condition~\eqref{itm:valid-2} would be satisfied.  Thus, $\rot(v_g)
  \le 0$ follows for this case.  Consider the decomposition of $\pi(s,
  t)$ into maximal subpaths consisting of semi-critical and normal
  edges.  If follows that each subpath consisting of semi-critical and
  normal edges hat rotation at most~0 and~$-1$, respectively.
  Moreover, the rotation at vertices between two subpaths is~0.
  Hence, if $\pi(s, t)$ contains at least one edge that is not
  semi-critical, we again get $\rot(\pi(s, t)) \le -1$ and thus
  Equation~\eqref{eq:rot-pi-st} is satisfied.  On the other hand, if
  $\pi(s, t)$ consists of semi-critical edges, we get the weaker
  inequality $\rot(\pi(s, t)) \le 0$.  If $\deg(s) = \deg(t) = 1$
  holds, Equation~\eqref{eq:rot-pi-st} is still satisfied as we have
  to show a weaker inequality in this case.  Otherwise, one of the
  poles has degree at least~2 and thus the edges incident to it cannot
  be semi-critical by definition.  Thus, the path $\pi(s, t)$ cannot
  consist of semi-critical edges.
  
  Finally, we allow \textbf{critical edges, i.e., $k \ge 0$}.  If
  $\pi(s, t)$ contains critical edges, we first consider these edges
  to have flexibility~1, leading to Equation~\eqref{eq:rot-pi-st} with
  $k = 0$.  Replacing an edge with flexibility~1 by an edge with
  flexibility~0 increases the rotation along $\pi(s, t)$ by at most~1.
  As $\pi(s, t)$ contains at most $k$ critical edges, $\rot(\pi(s,
  t))$ is increased by at most $k$ yielding
  Equation~\eqref{eq:rot-pi-st}.

  In the case that $\deg(s) = \deg(t) = 1$, the equation $\rot(\pi(s,
  t)) = -\rot(\pi(t, s))$ holds.  Equation~\eqref{eq:rot-pi-st}
  together with the inequality in the statement of the lemma leads to
  $k \ge \beta_\low + k + 1$, which is a contradiction.  In the
  following, we only consider the case where $\deg(s) = \deg(t) = 1$
  does not hold.  Since the total rotation around the outer face sums
  up to~$-4$, we get the following equation.
  \begin{equation*}
    \rot(\pi(s, t)) + \rot(\pi(t, s)) + \rot(s_g) + \rot(t_g) = -4
  \end{equation*}
  Recall that $\rot(s_g) = \sigma - 3$ and $\rot(t_g) = \tau - 3$.
  Using Equation~\eqref{eq:rot-pi-st} ($\deg(s) = \deg(t) = 1$ does
  not hold) and the inequality given in the lemmas precondition, we
  obtain the following.
  \begin{align}
    \nonumber&&\Big(k - 1\Big) - \Big(\overbrace{\left\lceil\frac{\sigma +
          \tau}{2}\right\rceil - 1}^{\beta_\low} + k + 1\Big) +
    \Big(\sigma - 3\Big) + \Big(\tau - 3\Big) &\ge -4 \\
    \nonumber&\Leftrightarrow\hspace{-1em}& -\left\lceil\frac{\sigma +
        \tau}{2}\right\rceil + (\sigma + \tau) &\ge 3 \\
    \label{eq:contradiction}
    &\Leftrightarrow\hspace{-1em}& \left\lfloor\frac{\sigma +
        \tau}{2}\right\rfloor &\ge 3
  \end{align}
  Recall that $\sigma + \tau \le 5$ is a requirement of the lemma.
  Thus, Equation~\eqref{eq:contradiction} is a contradiction, which
  concludes the proof of Claim~\ref{claim-1}.
  
  \begin{nrclaim}
    \label{claim-2}
    The valid cycle $C^\star$ exists.
  \end{nrclaim}

  Let $e^\star$ be the valid edge existing due to
  Claim~\ref{claim-1}.  If $e^\star$ is a loop with $s$ to its left
  and $t$ to its right, then $C^\star = e^\star$ is the desired valid
  cycle.  This case will serve as base case for a structural
  induction.

  Let $e^\star = (g, f)$ be a valid edge dual to $e$ having the outer
  face $g$ as source.  As $e^\star$ is not a loop, the graph $G - e$
  is still connected and thus $s$ and $t$ are contained in the same
  block of the graph $G - e + st$.  Let $H$ be this block (without
  $st$) and let $\mathcal S$ be the orthogonal representation of $H$
  induced by $\mathcal R$.  Then $H$ is a $k$-critical $st$-graph, as
  $H$ is a subgraph of $G$ and $H + st$ is biconnected.  Moreover, the
  path $\pi(t, s)$ is completely contained in $H$ and thus its
  rotation does not change.  Hence, all conditions for
  Lemma~\ref{lem:existence-of-cycle} are satisfied and since $H$
  contains fewer edges than $G$, we know by induction that there exists
  a valid cycle $C^\star_H$ such that bending $\mathcal S$ along
  $C^\star_H$ reduces $-\rot(\pi(t, s))$ by~1.  As the dual graph
  $H^\star$ of $H$ can be obtained from $G^\star$ by first contracting
  $e^\star$ and then taking a subgraph, all edges contained in
  $H^\star$ were already contained in $G^\star$.  Moreover, all valid
  edges in $H^\star$ are also valid in $G^\star$ and thus each edge in
  $C^\star_H$ corresponds to a valid edge in $G^\star$.  If these
  valid edges form a cycle in $G^\star$, then this is the desired
  cycle $C^\star$.  Otherwise, one of the two edges in $C^\star_H$
  incident to the outer face of $H$ is in $G^\star$ incident to the
  outer face $g$ of $G$ and the other is incident to the face $f$ of
  $G$.  In this case the edges of $C^\star_H$ from in $G^\star$ a path
  from $f$ to $g$ and thus adding the edge $e^\star$ yields the cycle
  $C^\star$, which concludes the proof of Claim~\ref{claim-2} and thus
  of this lemma.
\end{proof}

We get the following slightly weaker result for the case $\sigma =
\tau = 3$.

\begin{lemma}
  \label{lem:existence-of-cycle-special-cases}
  Let $G$ be a $k$-critical $st$-graph and let $\mathcal R$ be a
  valid $(3, 3)$-orthogonal representation. If $-\rot(\pi(t, s)) \ge
  \beta_{\low} + k + 2$ holds, then there exists a valid cycle $C^\star$
  such that bending $\mathcal R$ along $C^\star$ reduces $-\rot(\pi(t,
  s))$ by~1.
\end{lemma}
\begin{proof}
  Since $\sigma = \tau = 3$ holds, we have $\beta_\low = 2$ and thus
  $-\rot(\pi(t, s)) \ge k + 4$.  We add an edge $e = ss'$ with
  flexibility~1 to $G$, where~$s'$ is a new vertex, and consider the
  orthogonal representation $\mathcal R'$ of $G + e$ where $e$ has one
  bend such that $e$ contributes a rotation of~1 to $\pi(t, s')$.
  Since the rotation at $s$ in the outer face is~1, we have
  $\rot(\pi(t, s')) = \rot(\pi(t, s)) + 2$.  If follows that
  $-\rot(\pi(t, s')) \ge k + 4 - 2 = k + 2$ holds.  Since $\mathcal
  R'$ is a $(1, 3)$ orthogonal representation of $G + e$, and since
  the lower bound $\beta_\low'$ is~1 for $(1, 3)$ orthogonal
  representations, the precondition of
  Lemma~\ref{lem:existence-of-cycle}, namely the inequality
  $-\rot(\pi(t, s')) \ge \beta_\low' + k + 1$, is satisfied, which
  concludes the proof.
\end{proof}

The previous lemmas basically show that the existence of a valid
orthogonal representation with a lot of bends implies the existence of
valid orthogonal representations for a ``large'' interval of bend
numbers.  This is made more precise in the following.

Let $\mathcal B^\sigma_\tau$ be the set containing an integer $\beta$
if and only if $G$ admits a valid $(\sigma, \tau)$-orthogonal
representation with $\beta$ bends.  Assume $G$ admits a valid
$(\sigma, \tau)$-orthogonal representation, that is $\mathcal
B^\sigma_\tau$ is not empty.  We define the \emph{maximum bend value}
$\beta_{\max}$ to be the maximum in $\mathcal B^\sigma_\tau$.
Moreover, let $\beta \in \mathcal B^\sigma_\tau$ be the smallest
value, such that every integer between $\beta$ and $\beta_{\max}$ is
contained in~$\mathcal B^\sigma_\tau$.  Then we call the interval
$[\beta_\low, \beta-1]$ the \emph{$(\sigma, \tau)$-gap} of $G$.  The
value $\beta - \beta_\low$ is also called the \emph{$(\sigma,
  \tau)$-gap} of $G$; see Fig.~\ref{fig:gaps-1}.

\begin{figure}[tb]
  \centering
  \includegraphics[page=1]{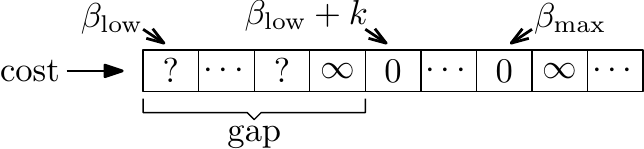}
  \caption{A cost function with gap~$k$.}
  \label{fig:gaps-1}
\end{figure}

\begin{lemma}
  \label{lem:small-gap}
  The $(\sigma, \tau)$-gap of a $k$-critical $st$-graph $G$ is at
  most~$k$ if $\sigma + \tau \le 5$.  The $(3, 3)$~gap of $G$ is at
  most $k + 1$.
\end{lemma}
\begin{proof}
  In the following, we assume $\sigma + \tau \le 5$; the case $\sigma
  = \tau = 3$ works literally the same when replacing
  Lemma~\ref{lem:existence-of-cycle} by
  Lemma~\ref{lem:existence-of-cycle-special-cases}.  Let $\mathcal R$
  be a valid $(\sigma, \tau)$-orthogonal representation with $\beta
  \ge \beta_\low + k + 1$ bends.  We show the existence of a valid
  $(\sigma, \tau)$-orthogonal representation with $\beta - 1$ bends.
  It follows that the number of bends can be reduced step by step down
  to $\beta_\low + k$, which shows that the gap is at most $k$.

  As $\mathcal R$ has $\beta$ bends, either $-\rot(\pi(s, t)) = \beta$
  or $-\rot(\pi(t, s)) = \beta$.  Without loss of generality, we
  assume $-\rot(\pi(t, s)) = \beta \ge \beta_\low + k + 1$.  Due to
  Lemma~\ref{lem:existence-of-cycle} there exists a valid cycle
  $C^\star$, such that bending along $C^\star$ reduces $-\rot(\pi(t,
  s))$ by~1.  This also reduces the number of bends by~1 (and thus
  yields the desired orthogonal representation) if $-\rot(\pi(s, t))$
  is not increased above $\beta - 1$.  Assume for a contradiction that
  $-\rot(\pi(s, t))$ was increased above $\beta - 1$.  Then in the
  resulting orthogonal representation $-\rot(\pi(s, t))$ is greater
  than~$\beta_\low$ and $-\rot(\pi(t, s))$ is at least~$\beta_\low$.
  It follows, that every $(\sigma, \tau)$-orthogonal representation
  has more than $\beta_\low$ bends, which contradicts the fact, that
  $\beta_\low$ is a tight lower bound.  
\end{proof}

The following lemma basically expresses the gap of an $st$-graph in
terms of the rotation along $\pi(s, t)$ instead of the number of
bends.

\begin{lemma}
  \label{lem:small-gap-rot}
  Let $G$ be an $st$-graph with $(\sigma, \tau)$-gap $k$.  The set
  $\{\rho \mid G$ admits a valid $(\sigma,\tau)$-orthogonal
  representation with $\rot(\pi(s, t)) = \rho\}$ is the union of at
  most $k+1$ intervals.
\end{lemma}
\begin{proof}
  Recall that an orthogonal representation of $G$ has $\beta$ bends if
  either $-\rot(\pi(s, t)) = \beta$ or $-\rot(\pi(t, s)) = \beta$.  We
  first consider the case that $-\rot(\pi(s, t)) = \beta$ for any
  number of bends $\beta \in [\beta_\low, \beta_{\max}]$.
  
  By the definition of the gap, there exists a valid orthogonal
  representation for $-\rot(\pi(s, t)) \in [\beta_{\low} + k,
  \beta_{\max}]$, which forms the first interval.  Moreover, $G$ does
  not admit a valid orthogonal representation with $\beta_{\low} + k -
  1$ bends, since the gap would be smaller otherwise.  Thus it remains
  to cover all allowed values contained in $[\beta_\low, \beta_\low +
  k - 2]$ by intervals.  In the worst case, exactly every second value
  is possible.  As $[\beta_\low, \beta_\low + k - 2]$ contains $k - 1$
  integers, this results in $\lceil(k - 1)/2\rceil$ intervals of size
  $1$.  Thus, we can cover all allowed values for $\rot(\pi(s, t))$ in
  case $-\rot(\pi(s, t)) \in [\beta_{\low} + k, \beta_{\max}]$ holds
  using only $\lceil(k - 1)/2\rceil + 1$ intervals.

  It remains to consider the case where $G$ has $\beta$ bends due to
  the fact that $-\rot(\pi(t, s)) = \beta$ holds.  With the same
  argument we can cover all possible values of $\pi(t, s)$ using
  $\lceil(k - 1)/2\rceil + 1$ intervals.  As $\rot(\pi(s, t))$ equals
  $-\rot(\pi(t, s))$ shifted by some constant, we can cover all
  allowed values for $\rot(\pi(s, t))$ using $2\cdot \lceil(k -
  1)/2\rceil + 2$ intervals.  If $k - 1$ is even, this evaluates to $k
  + 1$ yielding the statement of the lemma.  If $k - 1$ is odd and we
  assume the above described worst case, then we need one additional
  interval.  However, in this case there must exists a valid
  orthogonal representation with $\beta_\low$ bends and we counted two
  intervals for this bend number, namely for the case $-\rot(\pi(s,
  t)) = \beta_\low$ and $-\rot(\pi(t, s)) = \beta_\low$.  We show that
  a single interval suffices to cover both cases by showing that
  either $-\rot(\pi(s, t)) = \beta_\low$ or $-\rot(\pi(s, t)) =
  \beta_\low - 1$ holds if $-\rot(\pi(t, s)) = \beta_\low$.  This
  again leads to the desired $k + 1$ intervals.

  Due to the fact that the rotation around the outer face is~$-4$, the
  equation $-\rot(\pi(s, t)) = \sigma + \tau - 2 + \rot(\pi(t, s))$
  holds.  For $-\rot(\pi(t, s)) = \beta_\low$ we get the following.
  \begin{align*}
    \sigma + \tau - 2 - \beta_\low = \sigma + \tau - 2 -
    \left\lceil\frac{\sigma + \tau}{2}\right\rceil + 1
    =\left\lfloor\frac{\sigma + \tau}{2}\right\rfloor - 1
  \end{align*}
  If $\sigma + \tau$ is even, this is equal to $\beta_\low$, otherwise
  it is equal to $\beta_\low - 1$, which concludes the proof.
\end{proof}

\subsection{Computing the Cost Functions of Compositions}
\label{sec:comp-cost-funct}

Let $G$ be a graph with fixed planar embedding.  We describe a flow
network, similar to the one by Tamassia~\cite{t-eggmb-87} that can be
used to compute orthogonal representations of graphs with thick edges.
In general, we consider a flow network to be a directed graph with a
lower and an upper bound assigned to every edge and a demand assigned
to every vertex.  The bounds and demands can be negative.  An
assignment of flow-values to the edges is a feasible flow if it
satisfies the following properties.  The flow-value of each edge is at
least its lower and at most its upper bound.  For every vertex the
flow on incoming edges minus the flow on outgoing edges must equal its
demand.

We define the flow network $N$ as follows. The network $N$ contains a
node for each vertex of $G$, the \emph{vertex nodes}, each face of
$G$, the \emph{face nodes}, and each edge of $G$, the \emph{edge
  nodes}.  Moreover, $N$ contains arcs from each vertex to all
incident faces, the \emph{vertex-face arcs}, and similarly from each
edge to both incident faces, the \emph{edge-face arcs}.  We interpret
an orthogonal representation $\mathcal R$ of $G$ as a flow in $N$.  A
rotation $\rot(e_f)$ of an edge $e$ in the face $f$ corresponds to the
same amount of flow on the edge-face arc from $e$ to $f$.  Similarly,
for a vertex $v$ incident to $f$ the rotation $\rot(v_f)$ corresponds
to the flow from $v$ to $f$.

Obviously, the
properties~\eqref{prop:1-extended}--\eqref{prop:4-extended} of an
orthogonal representation are satisfied if and only if the following
conditions hold for the flow (note that we allow $G$ to have thick
edges).
\begin{compactenum}[(1)]
\item \label{prop:1-flow}The total amount of flow on arcs incident to
  a face node is~$4$ ($-4$ for the outer face).
\item \label{prop:2-flow}The flow on the two arcs incident to an edge
  node stemming from a $(\sigma, \tau)$-edge sums up to~$2 - (\sigma +
  \tau)$.
\item \label{prop:3-flow}The total amount of flow on arcs incident to
  a vertex node, corresponding to the vertex $v$ with incident edges
  $e_1, \dots, e_\ell$ occupying $\sigma_1, \dots, \sigma_\ell$
  incidences of $v$ is $\sum(\sigma_i + 1) - 4$.
\item \label{prop:4-flow}The flow on vertex-face arcs lies in the
  range $[-2, 1]$.
\end{compactenum}
Properties~\eqref{prop:1-flow}--\eqref{prop:3-flow} are equivalent to
the flow conservation requirement when setting appropriate demands.
Moreover, property~\eqref{prop:4-flow} is equivalent to the capacity
constraints in a flow network when setting the lower and upper bounds
of vertex-face arcs to~$-2$ and~$1$, respectively.  In the following,
we use this flow network to compute the cost function of a rigid
composition of graphs.  The term $T_{\flow}(\ell)$ denotes the time
necessary to compute a maximal flow in a planar flow network of size
$\ell$.

\begin{lemma}
  \label{lem:r-composition-fpt}
  The (restricted) cost functions of a rigid composition of $\ell$
  graphs can be computed in $O(2^k \cdot T_{\flow}(\ell))$ time if the
  resulting graph is $k$-critical.
\end{lemma}
\begin{proof}
  First note that in case of a rigid composition, computing
  ``restricted'' cost functions makes only a difference for the poles
  of the skeleton (as all other vertices have degree at least~3).
  However, enforcing $90^\circ$ angles for the poles is already
  covered by the number of incidences the resulting graph occupies at
  its poles.

  Let $H$ be the skeleton of the rigid composition of the graphs $G_1,
  \dots, G_\ell$ and let $G$ be the resulting graph with poles $s$ and
  $t$.  Before we show how to compute orthogonal representations of
  $G$, we show that the number of incidences $\sigma_i$ and $\tau_i$ a
  subgraph $G_i$ occupies at its poles $s_i$ and $t_i$ is (almost)
  fixed.  Assume $s_i$ is not one of the poles $s$ or $t$ of $G$.
  Then $s_i$ has at least three incident edges in the skeleton $H$ as
  $H + st$ is triconnected.  Thus, the subgraph $G_i$ occupies at
  most two incidences in any orthogonal representation of $G$, and hence
  $s_i$ has either degree~1 or degree~2 in $G_i$.  In the former case
  $\sigma_i$ is~1, in the latter $\sigma_i$ has to be~2.  If $s_i$ is
  one of the poles of $G$, then it may happen that $G_i$ occupies
  incidences in some orthogonal representations of $G$ and three
  incidences in another orthogonal representation.  However, this
  results in a constant number of combinations and thus we can assume
  that the values $\sigma_i$ and $\tau_i$ are fixed for $i \in \{1,
  \dots, \ell\}$.

  To test whether $G$ admits a valid $(\sigma, \tau)$-orthogonal
  representation, we can instead check the existence of a valid
  orthogonal representation of $H$ using thick edges for the graphs
  $G_1, \dots, G_\ell$ (more precisely, we use a $(\sigma_i,
  \tau_i)$-edge for $G_i$).  To ensure that substituting the thick
  edges with the subgraphs yields the desired orthogonal
  representation, we have to enforce the following properties for the
  orthogonal representation of $H$.  First, the orthogonal
  representation of $H$ has to occupy $\sigma$ and $\tau$ incidences
  at its poles.  Second, the thick edge corresponding to a subgraph
  $G_i$ is allowed to have $\beta_i$ bends only if $G_i$ has a valid
  $(\sigma_i, \tau_i)$-orthogonal representation with $\beta_i$ bends.
  Note that this tests the existence of an orthogonal representation
  without restriction to the number of bends.  We will show later, how
  to really compute the cost function of $G$.

  Restricting the allowed flows in the flow network such that they
  only represent $(\sigma, \tau)$-orthogonal representations is easy.
  The graph $H$ occupies $\sigma$ incidences if and only if $\rot(s_f)
  = \sigma - 3$ (where $f$ is the outer face).  As the rotation
  $\rot(s_f)$ is represented by the flow on the corresponding
  vertex-face arc, we can enforce $\rot(s_f) = \sigma - 3$ by setting
  the upper and lower bound on the corresponding arc to $\sigma - 3$.
  Analogously, we can ensure that $H$ occupies $\tau$ incidences of
  $t$.

  In the following we show how to restrict the number of bends of a
  thick edge $e_i = s_it_i$ to the possible number of bends of
  the subgraph $G_i$ it represents.  Assume $G_i$ is $k_i$-critical.
  It follows from Lemma~\ref{lem:small-gap} that $G_i$ has gap at most
  $k_i$.  Thus, the possible values for $\rot(\pi(s_i, t_i))$ can be
  expressed as the union of at most $k_i + 1$ intervals due to
  Lemma~\ref{lem:small-gap-rot}.  Restricting the rotation to an
  interval can be easily done using capacities.  However, we get $k_i
  + 1$ possibilities to set these capacities, and thus combining these
  possibilities for all thick edges results in $\prod (k_i + 1)$ flow
  networks.

  We show that $\prod (k_i + 1)$ is in $O(2^k)$.  To this end, we
  first show that $\sum k_i \le k$ holds, by proving that an edge
  that is critical in one of the subgraphs $G_i$ is still critical in
  the graph $G$.  This is obviously true for critical edges in $G_i$
  not incident to a pole of $G_i$, as these inflexible edges already
  have endpoints with degree~4 in $G_i$.  An edge $e$ incident to a
  pole, without loss of generality $s_i$ of $G_i$ is critical in $G_i$
  if $s_i$ has degree at least~2.  If $s_i$ remains a pole of $G$,
  then $e$ is also critical with respect to $G$.  Otherwise, $s_i$ has
  degree~4 in $G$, which comes from the fact that the skeleton $H$
  becomes triconnected when adding the edge $st$.

  As the $0$-critical subgraphs do not play a role in the product
  $\prod(k_i + 1)$, we only consider the $d$ subgraphs $G_1, \dots,
  G_d$ such that $G_i$ (for $i \in \{1, \dots, d\}$) is $k_i$-critical
  with $k_i \ge 1$.  To find the worst case, we want to maximize
  $\prod(k_i + 1)$ with respect to $\sum k_i \le k$ (which is
  equivalent to finding a hypercuboid of dimension $d$ with maximal
  volume and with fixed perimeter).  We get the maximum by setting
  $k_i = k/d$ for all subgraphs, which results in $(k/d + 1)^d$
  combinations.  Substituting $k/d = x$ leads to $x^{k/x}$, which
  becomes maximal, when $x^{1/x}$ is maximal.  Since $f(x) = x^{1/x}$
  is a decreasing function, we get the worst case for $x = 1$ (when
  restricting $x$ to positive integers), which corresponds to $d = k$
  graphs that are 1-critical.  Thus, in the worst case, we get
  $O(2^k)$ different combinations.

  Since the flow networks have size $O(\ell)$, we can test the
  existence of a valid orthogonal representation of $G$ in $O(2^k\cdot
  T_{\flow}(\ell))$ time.  However, we want to compute the cost
  function instead.  Assume we want to test the existence of a valid
  orthogonal representation with a fixed number of bends $\beta$.  In
  the following, we show how to restrict each of the flow networks to
  allow only flows corresponding to orthogonal representation with
  $\beta$ bends.  Then $G$ clearly admits a valid orthogonal
  representation with $\beta$ bends if and only if one of these flow
  networks admits a valid flow.  The orthogonal representation of $H$
  (and thus the resulting one of $G$) has $\beta$ bends if either
  $-\rot(\pi(s, t)) = \beta$ or $-\rot(\pi(t, s)) = \beta$.  We can
  consider these two cases separately, resulting in a constant factor
  in the running time.  Thus, it remains to ensure that $-\rot(\pi(s,
  t))$ is fixed to $\beta$.  This can be done by splitting the face
  node corresponding to the outer face such that exactly the arcs
  entering $f$ from edge nodes or vertex nodes corresponding to edges
  and internal vertices of $\pi(s, t)$ are incident to one of the
  resulting nodes.  Restricting the flow between the two resulting
  nodes representing the outer face $f$ to $\beta$ obviously enforces
  that $-\rot(\pi(s, t)) = \beta$ holds.  Thus, we could get the cost
  function of $G$ by doing this for all possible values of $\beta$.
  However, we can get the cost function more efficiently.

  Instead of fixing the value of $-\rot(\pi(s, t))$ to $\beta$, we can
  compute maximum flows to minimize or maximize it.  Let $\rot_{\min}$
  and $\rot_{\max}$ be the resulting minimum and maximum for
  $-\rot(\pi(s, t))$, respectively.  Note that, if $\rot_{\max}$ is
  less than $\beta_\low$, then there is no orthogonal representation
  where the number of bends are determined by the rotation along
  $\pi(s, t)$.  Moreover, if $\rot_{\min} < \beta_\low$, we set
  $\rot_{\min} = \beta_\low$.  It follows from basic flow theory that
  all values between $\rot_{\min}$ and $\rot_{\max}$ are also
  possible.  Thus, after computing the two flows, we can simply set
  the cost function of $G$ to~0 on that interval.  To save a factor of
  $k$ in the running time we do not update the cost function of $G$
  immediately, but store the interval $[\rot_{\max}, \rot_{\min}]$.
  In the end, we have $O(2^k)$ such intervals.  The maximum of all
  upper bounds of these intervals is clearly $\beta_{\max}$ (the
  largest possible number of bends of $G$).  It remains to extract the
  cost function of $G$ on the interval $[\beta_\low, \beta_\low +
  k-1]$, since the cost function of $G$ has gap at most~$k$
  (Lemma~\ref{lem:small-gap}).  This can be done by sorting all
  intervals having their lower bound in $[\beta_\low, \beta_\low + k -
  1]$ by their lower bound.  This can be done in $O(k + 2^k)$ time,
  since we sort $O(2^k)$ values in a range of size $k$.  Finally, the
  cost function on $[\beta_\low, \beta_\low + k - 1]$ can be easily
  computed in $O(k + 2^k)$ time by scanning over this list.  As this
  is dominated by the computation of all flows, we get an overall
  running time of $O(2^k\cdot T_{\flow}(\ell))$.
\end{proof}

\begin{lemma}
  \label{lem:sp-composition-fpt}
  The (restricted) cost functions of a series and a parallel
  composition can be computed in $O(k^2 + 1)$ time if the resulting
  graph is $k$-critical.
\end{lemma}
\begin{proof}
  First, consider only the non-restricted case.  Let $G_1$ and $G_2$
  be the two graphs that should be composed and let $G$ be the
  resulting graph.  As in the rigid case, we can use flow networks to
  compute the cost functions of $G$.  However, this time the flow
  network has constant size and thus we do not have to be so careful
  with the constants.

  Assume $G_1$ and $G_2$ are $k_1$- and $k_2$-critical, respectively.
  Up to possibly a constant number, all critical edges in $G_i$ are
  also critical in $G$, that is $k_i \in O(k + 1)$ (note that the
  ``$+1$'' is necessary for the case $k = 0$).  Thus, both graphs
  $G_1$ and $G_2$ have a gap of size $O(k + 1)$.  It follows that the
  possible rotations values for $\pi(s_i, t_i)$ (where $s_i$ and $t_i$
  are the poles of $G_i$) are the union of $O(k + 1)$ intervals, which
  results in $O(k^2 + 1)$ possible combinations and thus $O(k^2 + 1)$
  flow networks of constant size.  Note that we get an additional
  constant factor by considering all possible values for the number of
  occupied incidences of the graphs $G_i$.  Extracting the cost
  functions out of the results from the flow computation can be done
  analogously to the case where we had a rigid composition (proof of
  Lemma~\ref{lem:r-composition-fpt}), which finally results in the
  claimed running time $O(k^2 + 1)$.

  To compute the restricted cost functions, one possibly has to
  restrict the rotation at some vertices to $-1$ or $1$, which can be
  obviously done without increasing the running time.
\end{proof}

\begin{theorem}
  {\sc FlexDraw} for $k$-critical graphs can be solved in $O(2^k\cdot
  n \cdot T_{\flow}(n))$
\end{theorem}
\begin{proof}
  By Theorem~\ref{thm:general-algo}, we get an algorithm with the
  running time $O(n \cdot (n \cdot T_S + n\cdot T_P + T_R(n)))$, where
  $T_S, T_P \in O(k^2 + 1)$ (Lemma~\ref{lem:sp-composition-fpt}) and
  $T_R(\ell) = 2^k\cdot T_{\flow}(\ell)$
  (Lemma~\ref{lem:r-composition-fpt}) holds.  This obviously yields
  the running time $O((k^2+1)\cdot n^2 + 2^k\cdot n \cdot
  T_{\flow}(n)) = O(2^k\cdot n \cdot T_{\flow}(n))$.
\end{proof}

\section{Conclusion}
\label{sec:conclusion}

We want to conclude with the open question whether there exists an
FPT-algorithm for {\sc OptimalFlexDraw} for the case where all cost
functions are convex and where the first bend causes cost only for $k$
edges (that is we have $k$ inflexible edges).  One might think that
this works similar as for {\sc FlexDraw} by showing that the cost
functions of $st$-graphs are only non-convex if they contain
inflexible edges.  Then, when encountering a rigid composition, one
could separate these non-convex cost functions into convex parts and
consider all combinations of these convex parts.  Unfortunately, the
cost functions of $st$-graphs may already be non-convex, even though
they do not contain inflexible edges.  The reason why {\sc
  OptimalFlexDraw} can still be solved efficiently if there are no
inflexible edges~\cite{brw-oogd-13} is that, in this case, the cost
functions need to be considered only up to three bends (and for this
restricted intervals, the cost functions are convex).  However, a
single subgraph with inflexible edges in a rigid composition may force
arbitrary other subgraphs in this composition to have more than three
bends, potentially resulting in linearly many non-convex cost
functions that have to be considered.  Thus, although the algorithms
for {\sc FlexDraw} and {\sc OptimalFlexDraw} are very similar, the
latter does not seem to allow even a small number of inflexible edges.

\medskip
\noindent\textbf{Acknowledgement.\,\,} We thank Marcus Krug for
discussions on \textsc{FlexDraw}.

\bibliographystyle{splncs03}
\bibliography{strings,all}

\end{document}